\documentclass[runningheads]{llncs}

\usepackage{newclude}
\usepackage{graphicx}
\usepackage{amsfonts}
\usepackage{enumerate}
\usepackage[scr]{rsfso}
\usepackage{units}
\usepackage{enumitem}
\usepackage{comment}
\usepackage{color}
\usepackage{textcomp}
\usepackage{hyperref}
\usepackage{multirow}
\usepackage{environ}
\usepackage[boxed]{algorithm2e}
\usepackage{environ}
\usepackage{comment}
\usepackage{makecell}
\usepackage[classfont=roman,funcfont=roman,langfont=roman]{complexity}
\usepackage{amsmath}
\usepackage{cleveref}
\usepackage{todonotes}
\usepackage{tcolorbox} 
\usepackage[title]{appendix}

\newtheorem{rrule}{Reduction rule}

\renewcommand{\O}{\mathcal{O}}

\title{Tur\'{a}n's Theorem Through  Algorithmic Lens\thanks{The research leading to these results has received funding from the the Research Council of Norway via the project BWCA (grant no. 314528) and DFG Research Group ADYN via grant DFG 411362735.}}
\author{}

\newcommand{\polyn}{n^{\mathcal{O}(1)}}
\newcommand{\probTuranClique}{\textsc{Tur\'{a}n's Clique}\xspace}
\newcommand{\probTuranIS}{\textsc{Tur\'{a}n's Independent Set}\xspace}

\newcommand{\probBrookIS}{\textsc{Independent Set over Brook's bound}\xspace}

\usepackage{empheq}
\usepackage{tcolorbox}
\definecolor{blueish}{rgb}{0.122, 0.435, 0.698}
\definecolor{dagstuhlyellow}{rgb}{0.99,0.78,0.07}
\definecolor{lightgray}{rgb}{0.9,0.9,0.9}
\newtcbox{\colbox}{
	size=title,
	nobeforeafter,
	colframe=white,
	colback=lightgray,
	arc=10pt,
	tcbox raise base}

\newcommand{\defproblemu}[3]{
	\vspace{1mm}
	\noindent\fbox{
		\begin{minipage}{0.95\textwidth}
			#1 \\
			{\bf{Input:}} #2  \\
			{\bf{Question:}} #3
		\end{minipage}
	}
	\vspace{1mm}
}

\institute{Department of Informatics, University of Bergen, Norway. \email{\{fomin,petr.golovach\}@ii.uib.no} \and
	St.\ Petersburg Department of V.A.\ Steklov Institute of Mathematics, Russia.\\\email{danilka.pro@gmail.com} \and
	Hasso Plattner Institute, University of Potsdam, Germany.\\\email{kirillsimonov@gmail.com}
	}

\author{
	Fedor V. Fomin\inst{1}
	\and
	Petr A. Golovach\inst{1}
	\and
	Danil Sagunov\inst{2}
	\and
	Kirill Simonov\inst{3}
}

\authorrunning{F.~V.~Fomin et al.}

\begin{document}

	\maketitle
	
\begin{abstract}
The fundamental theorem of Tur\'{a}n  from Extremal Graph Theory determines the exact bound on the number of edges $t_r(n)$ in an $n$-vertex graph that does not contain a clique of size $r+1$. We establish an interesting link between Extremal Graph Theory and Algorithms by providing a simple compression algorithm that in linear time reduces the problem of finding a clique of size $\ell$ in an $n$-vertex graph $G$ with $m \ge t_r(n)-k$ edges, where $\ell\leq r+1$, to the problem of finding a maximum clique in a graph on at most $5k$ vertices. This also gives us an algorithm deciding in time  
 $2.49^{k}\cdot(n + m)$ whether $G$ has a clique of size $\ell$. 
  As a byproduct of the new compression algorithm,  we give an algorithm  that in time $2^{\mathcal{O}(td^2 )}  \cdot n^2$ decides whether a graph contains an independent set of size at least  $n/(d+1) +t$. Here $d$ is the average vertex degree of the graph $G$.  
  The multivariate complexity analysis  based on ETH indicates that the asymptotical dependence on several parameters in the running times of our algorithms is tight. 
 \end{abstract}

\section{Introduction}\label{sec:intro}
In 1941, P\'{a}l Tur\'{a}n published a theorem that became one of the central results in extremal graph theory. The theorem bounds the number of edges in an undirected graph that does not contain a complete subgraph of a given size.
For  positive integers $r\leq n$, the  \emph{Tur\'{a}n's graph} $T_r(n)$ is the unique complete $r$-partite $n$-vertex graph where each part consists of $\lfloor \frac{n}{r}\rfloor$ or $\lceil \frac{n}{r}\rceil$ vertices. In other words, 
	 $T_r(n)$ is isomorphic to  $K_{a_1,a_2,\ldots,a_r}$, where $a_i=\lceil\frac{n}{r}\rceil$ if $i$ is less than or equal to $n$ modulo $r$ and $a_i=\lfloor \frac{n}{r}\rfloor$ otherwise. We use $t_r(n)$  to denote the number of edges in $T_r(n)$.
	 
 \begin{theorem}[Tur\'{a}n's Theorem~\cite{MR18405}]\label{thm:turan}
	Let $r \le n$. Then any $K_{r+1}$-free $n$-vertex graph has at most $t_r(n)$ edges.
	The only $K_{r+1}$-free $n$-vertex graph with exactly $t_r(n)$ edges is $T_r(n)$.
\end{theorem}

The theorem yields a polynomial time algorithm  that for a given $n$-vertex graph $G$ with at least $t_r(n)$ edges decides whether $G$ contains a clique $K_{r+1}$. Indeed, if a graph $G$ is isomorphic to $T_r(n)$, which is easily checkable in polynomial time, then it has no clique of size ${r+1}$. Otherwise, by Tur\'{a}n's theorem, $G$ 
  contains $K_{r+1}$. There are   constructive proofs of Tur\'{a}n's theorem that also allow to find a clique of size $r+1$ in a graph with at least $t_r(n)$ edges.

 The fascinating question is whether Tur\'{a}n's theorem could help to find efficiently larger cliques in sparser graphs. There are two natural approaches to defining a ``sparser'' graph and a ``larger'' clique. These approaches bring us to the following questions; addressing these questions is the primary motivation of our work. 
   
 First, what happens when the input graph has a bit less edges than the Tur\'{a}n's graph? More precisely,
  \begin{tcolorbox}[colback=green!5!white,colframe=blue!40!black]
 Is there an efficient algorithm that for some $k\geq 1$, decides whether an $n$-vertex  graph with  at least $t_r(n)-k$ edges 
 contains a clique of size $r+1$?
 \end{tcolorbox}
 \vskip-2mm 
    
 Second, could Tur\'{a}n's theorem be useful in finding a clique of size  larger than $r+1$ in an $n$-vertex graph with $t_r(n)$ edges? That is, 
   
 \begin{tcolorbox}[colback=green!5!white,colframe=blue!40!black]
 Is there an efficient algorithm that for some $\ell> r$ decides whether an $n$-vertex  graph with  at least $t_r(n)$ edges 
 contains a clique of size $\ell$?
 \end{tcolorbox}
\vskip-2mm

 We provide answers to both questions, and more. We resolve the first question by showing a \emph{simple} fixed-parameter tractable (FPT) algorithm where the parameter is $k$, the ``distance'' to the Tur\'{a}n's graph.
 Our algorithm builds on the cute ideas used by Erd\H{o}s in his proof of Tur\'{a}n's theorem~\cite{MR307975}. Viewing these ideas through algorithmic lens leads us to a simple preprocessing procedure, formally a linear-time polynomial compression. For the second question, unfortunately, the answer is negative. 

\smallskip
\noindent\textbf{Our contribution.} To explain our results, it is convenient to state the above questions in terms of  the computational complexity of the following problem.

\defproblemu{\probTuranClique}{An $n$-vertex graph $G$,  positive integers   $r, \ell \le n$,  and $k$ such that $|E(G)|\ge t_r(n)-k$.}{Is there a clique of size at least $\ell$ in $G$?}

Our first result is the following theorem (\Cref{theorem:kernel}).
Let $G$ be an $n$-vertex graph  with $m \ge t_r(n)-k$ edges. Then there is an algorithm that 
 for any $\ell\leq r+1$, in time   $2.49^{k}\cdot(n + m)$ either finds 
  a clique of size at least $\ell$ in $G$ or correctly reports that $G$ does not have a clique of size $\ell$. Thus for $\ell\leq r+1$, 
  \probTuranClique 
  is FPT parameterized by $k$. More generally, we prove that the problem admits a  compression of size linear in $k$. That is, we provide a linear-time procedure that reduces an   instance $(G, r,\ell, k)$  of  \probTuranClique to an equivalent instance $(G',p)$ of the \textsc{Clique} problem with at most $5k$ vertices. The difference between \textsc{Clique} and  \probTuranClique is that we do not impose any bound on the number of edges in the input graph of \textsc{Clique}. This is why we use the term compression rather than kernelization,\footnote{A \emph{kernel} is by definition a reduction to an instance of the same problem. See the book \cite{fomin2019kernelization} for an introduction to kernelization.} and we argue that stating our reduction in terms of compression is far more natural and helpful.
   Indeed, after reducing the instance to the size linear in the parameter $k$, the difference between \textsc{Clique} and \probTuranClique vanes, as even the total number of edges in the instance is automatically bounded by a function of the parameter. On the other hand, \textsc{Clique} is a more general and well-studied problem than \probTuranClique.
 
   Pipelined with the fastest known exact algorithm for \textsc{Maximum Independent Set} of running time $\O(1.1996^n)$
\cite{XiaoN17}, our reduction provides the FPT algorithm for  \probTuranClique parameterized by $k$. 
This algorithm is single-exponential in $k$ and linear in $n+m$, and we also show that the existence of an algorithm subexponential in $k$ would contradict Exponential Time Hypothesis (\Cref{cor:single_tight}). Thus the running time of our algorithm is essentially tight, up to the constant in the base of the exponent.

The condition  $\ell\leq r+1$ required by our algorithm is, unfortunately, unavoidable. We prove (\Cref{thm:TuranNPC})  that for any  fixed $p\ge 2$, the problem of deciding whether an $n$-vertex graph  with at least  $t_r(n)$ edges contains a clique of size $\ell=r+p$ is \NP-complete. Thus for any  $p\ge 2$, \probTuranClique parameterized by $k$ is para-NP-hard. (We refer to the book of Cygan et al. \cite{cygan2015parameterized}  for an introduction to parameterized complexity.)

While our hardness result rules out finding cliques of size $\ell> r+1$ in graphs with $t_r(n)$ edges in FPT time, an interesting  
 situation arises when the ratio $\xi:= \lfloor\frac{n}{r}\rfloor $ is small. In the extreme case, when $n=r$,  the $n$-vertex graph $G$ with  $t_r(n)=n(n-1)/2$  is a complete graph. In this case the problem  becomes trivial. 

 
 To capture how far the desired clique is from the Tur\'{a}n's bound, we introduce the parameter
\[
\tau= \begin{cases}
      0, & \text{ if } \ell \leq r, \\
      \ell - r, & \text{otherwise}.
\end{cases}
\]
The above-mentioned compression algorithm into 
\textsc{Clique} with at most $5k$ vertices  yields almost ``for free''  a compression  of   \probTuranClique
 into 
\textsc{Clique} with  $\mathcal{O}(\tau\xi^2+k)$ vertices. Hence for any $\ell$, one can decide whether an  $n$-vertex   graph with $m \ge t_r(n)-k$ edges contains a clique of size $\ell$   in time  $2^{\mathcal{O}(\tau\xi^2+k)}  \cdot (n + m)$. Thus the problem is FPT parameterized by $\tau+\xi +k$. This result has an interesting interpretation  when  we 
look for a large independent set in the complement of a graph.
Tur\'{a}n's  theorem, when applied to the complement $\overline{G}$ of a graph $G$, yields a  bound
 \[
 \alpha({G})\geq \frac{n}{d + 1},
 \]
 where $\alpha(G)$ is the size of the largest independent set in $G$ (the independence number of $G$), and $d$ is the average vertex degree of $G$. 
 This motivates us to define the following problem. 
 
 \defproblemu{\probTuranIS}{An $n$-vertex graph $G$ with average degree $d$,  a positive integer $t$.}{Is there an independent set  of size at least $\frac{n}{d+1}+t$ in $G$?}
 
By \Cref{thm:tauxik},  we have a simple algorithm (\Cref{cor:comprBrookcompr}) that compresses an instance of 
\probTuranIS   into  an instance of  \textsc{Independent Set} with   $ \mathcal{O}(td^2 )$ vertices. Pipelined with an exact algorithm computing a maximum independent set, the compression results in the algorithm solving  \probTuranIS in time
    $2^{\mathcal{O}(td^2 )}  \cdot n^2$.

As we already mentioned, \probTuranClique is \NP-complete  for any fixed $\tau \ge 2$ and $k=0$. We prove that  the problem remains intractable being parameterized by any pair of the parameters from the triple  $\{\tau, \xi, k\}$.
More precisely,  \probTuranClique is also 
 \NP-complete for any fixed $\xi\ge 1$ and $\tau=0$, as well as for
 any fixed $\xi\ge 1$ and $k=0$. These lower bounds are given in \Cref{thm:TuranNPC}.

 Given the algorithm of running time  $2^{\mathcal{O}(\tau\xi^2+k)}  \cdot (n + m)$ and the lower bounds for parameterization by any pair of the parameters from $\{\tau, \xi, k\}$, a natural question is, what is the optimal dependence of a \probTuranClique algorithm on  $\{\tau, \xi, k\}$? We use the Exponential Time Hypothesis (ETH) of Impagliazzo,  Paturi,  and Zane \cite{ImpagliazzoPZ01} to address this question. Assuming ETH, we rule out  the existence of algorithms solving \probTuranClique in time 
 $f(\xi,\tau)^{o(k)}\cdot n^{f(\xi,\tau)}$,  $f(\xi,k)^{o(\tau)} \cdot n^{f(\xi, k)}$, and   $f(k,\tau)^{o(\sqrt{\xi})}\cdot n^{f(k, \tau)}$, for any function $f$ of the respective parameters.
 
\smallskip
\noindent\textbf{Related work.}
\textsc{Clique} is a notoriously difficult computational problem. It is one of Karp's 21 \NP-complete problems~\cite{Karp72}
and by the work of H\aa stad, it is hard to approximate \textsc{Clique} within a factor of
$n^{1-\epsilon}$~\cite{MR1687331}. \textsc{Clique} parameterized by the solution size is W[1]-complete \cite{DowneyF99}. The problem plays the fundamental role in the W-hierarchy of Downey and Fellows,  and serves as the starting point in the majority of parameterized hardness reductions. From the viewpoint of structural parameterized kernelization, \textsc{Clique} does not admit a polynomial kernel when parameterized by the size of the vertex cover~\cite{BodlaenderJK14}.  A notable portion of works in parameterized algorithms and  kernelization is devoted to solving \textsc{Independent Set} (equivalent to \textsc{Clique} on the graph's complement) on specific graph classes like planar, $H$-minor-free graphs and nowhere-dense graphs
  \cite{DemaineFHT05jacm,BodlaenderFLPST16,DvorakM17,PilipczukS21}.
   
  Our algorithmic study of  Tur\'{a}n's theorem  fits into the paradigm of the ``above guarantee'' parameterization~\cite{MahajanR99}. 
 This approach was 
 successfully applied to various problems, see e.g.~\cite{AlonGKSY10,CrowstonJMPRS13,GargP16,DBLP:journals/mst/GutinKLM11,GutinIMY12,GutinP16,GutinRSY07,LokshtanovNRRS14,MahajanRS09,DBLP:conf/wg/Jansen0N19,fomin_et_al:LIPIcs:2019:11168}.

Most relevant to our work is the 
 work of Dvorak and Lidicky on independent set ``above Brooks' theorem'' \cite{DvorakL16}.  By Brooks' theorem~\cite{Brooks41}, every $n$-vertex graph of maximum degree at most $\Delta\geq  3$ and clique number at most $\Delta$  has an independent set of size at least $n/\Delta$. Then the 
 \probBrookIS problem
 is to decide whether an input graph $G$ has an independent set of size at least  $\frac{n}{\Delta} +p$. 
  Dvorak and Lidicky \cite[Corollary~3]{DvorakL16} proved  that \probBrookIS  admits a kernel with at most $114p\Delta^3  $ vertices.  This  kernel  also implies an algorithm of running time   $2^{\mathcal{O}(p\Delta^3 )}  \cdot\polyn$. 
 When average degree $d$ is at most $\Delta -1$,   by \Cref{cor:comprBrookcompr}, we have that  \probBrookIS admits a   compression into an instance of \textsc{Independent Set} with  $\mathcal{O}(p\Delta^2)$ vertices. Similarly, 
 by  \Cref{cor:comprBrookalgo}, for  $d \leq \Delta  -1$, \probBrookIS is solvable in time $2^{\mathcal{O}(p\Delta^2 )}  \cdot\polyn$.  When  $d> \Delta -1$,   for example, on regular graphs, the result of Dvorak and Lidicky is non-comparable with our results.

	\section{Algorithms}\label{sec:algorithms}
	%
	While in the literature it is common to present Tur\'{a}n's theorem under the implicit assumption that $n$ is divisible by $r$,
	here we make no such assumption. For that, it is useful to recall the precise value of $t_r(n)$ in the general setting, as observed by Tur\'{a}n~\cite{MR18405}.
	
	\begin{proposition}[Tur\'{a}n~\cite{MR18405}]\label{lemma:turans_edges}
		For   positive integers $r \le n$, 
		$$t_r(n)=\left(1-\frac{1}{r}\right)\cdot\frac{n^2}{2}-\frac{s}{2}\cdot\left(1-\frac{s}{r}\right)$$
		where $s=n-r\cdot\lfloor\frac{n}{r}\rfloor$ is the remainder in the division of $n$ by $r$.
	\end{proposition}
	Note that \cite{MR18405} uses the expression $t_r(n) = \frac{r - 1}{2r} \cdot (n^2 - s^2) + \binom{s}{2}$, however it can be easily seen to be equivalent to the above.
	We start with our main problem, where we look for a $K_{r + 1}$ in a graph that has slightly less than $t_r(n)$ edges.
	Later in this section, we show how to derive our other algorithmic results from the compression routine developed next.
	
	\subsection{Compression algorithm for $\ell \le r + 1$}
	
	First, we make a crucial observation on the structure of a \probTuranClique instance that will be the key part of our compression argument.
	Take a vertex $v$ of maximum degree in $G$, partition $V(G)$ on $S = N_G(v)$ and $T = V(G) \setminus S$, and add all edges between $S$ and $T$ while removing all edges inside $T$. It can be argued that this operation does not decrease the number of edges in $G$ while also preserving the property of being $K_{r + 1}$-free. Performing this recursively yields that $T_r(n)$ has indeed the maximum number of edges for a $K_{r + 1}$-free graph, and this is the gist of Erd\H{o}s' proof of Tur\'{a}n's Theorem~\cite{MR307975}. Now, we want to extend this argument to cover our above-guarantee case. Again, we start with the graph $G$ and perform exactly the same recursive procedure to obtain the graph $G'$. While we cannot say that $G'$ is equal to $G$, since the latter has slightly less than $t_r(n)$ edges, we can argue that every edge that gets changed from $G$ to $G'$ can be attributed to the ``budget'' $k$. Thus we arrive to the conclusion that $G$ is different from $G'$ at only $\mathcal{O}(k)$ places.
	The following lemma makes this intuition formal.
	
	
	\begin{lemma}\label{lemma:partition}
		There is an $\mathcal{O}(m + k)$-time algorithm that  for   non-negative integers $k \ge 1$, $r \ge 2$ and an $n$-vertex graph $G$ with $m \ge t_r(n)-k$ edges, finds a partition $V_1, V_2, \ldots, V_p$ of $V(G)$ with the following properties
		\begin{enumerate} 
			\item[$(i)$] $p \ge r-k$;
			\item[$(ii)$] For each $i\in \{1,\dots, p\}$, there is a vertex $v_i \in V_i$ with $N_G(v_i) \supset  V_{i+1}\cup V_{i+2} \cup \cdots \cup V_p$;
			\item[$(iii)$] If $p \le r$, then for the complete $p$-partite graph $G'$ with parts $V_1, V_2, \ldots, V_p$, we have $|E(G')|\ge |E(G)|$ and $|E(G)\triangle E(G')|\le 3k$.
			Moreover, all vertices covered by $E(G)\setminus E(G')$ are covered by $E(G')\setminus E(G)$ and $|E(G')\setminus E(G)|\le 2k$.
		\end{enumerate}
	\end{lemma}

	Let us clarify this technical definition.
	The lemma basically states that if a graph $G$ has at least $t_r(n)-k$ edges, then it either has a clique of size $r+1$, or it has at most $3k$ edit distance to a complete multipartite graph $G'$ consisting of $p \in [r-k,r]$ parts.
	Moreover, $G$ has a clique of size $p$ untouched by the edit, i.e.\ this clique is present in the complete $p$-partite graph $G'$ as well.
	
	We should also note that \Cref{lemma:partition} is close to the concept of stability of Tur\'{a}n's theorem.
	This concept received much attention in extremal graph theory (see e.g.\ recent work of Kor\'{a}ndi et al.\ \cite{korandi2021exact}), and appeals the structural properties of graphs having number of edges close to the Tur\'{a}n's number $t_r(n)$.
	\Cref{lemma:partition} can also be seen as a stability version of Tur\'{a}n's theorem, but from the algorithmic point of view.
	We move on to the proof of the lemma.
	
	\begin{proof}[of \Cref{lemma:partition}]
		First, we state the algorithm, which follows from the  Erd\H{o}s' proof of Tur\'{a}n's Theorem from~\cite{MR307975}.
		We start with an empty graph $G'$ defined on the same vertex set as $G$, and set $G_1 = G$.
		Then we select the vertex $v_1 \in V(G_1)$ as an arbitrary maximum-degree vertex in $G_1$, i.e.\ $\deg_{G_1}(v_1)=\max_{u\in V(G_1)}\deg_{G_1}(u)$.
		We put $V_1=V(G_1)\setminus N_{G_1}(v_1)$ and add to $G'$ all edges between $V_1$ and $V(G_1)\setminus V_1$.
		
		We then put $G_2:=G_1-V_1$ and, unless $G_2$ is empty, apply the same process to $G_2$.
		That is, we select $v_2 \in V(G_2)$ with $\deg_{G_2}(v_2)=\max_{u \in  V(G_2)}\deg_{G_2}(u)$ and put $V_2=V(G_2)\setminus N_{G_2}(v_2)$ and add all edges between $V_2$ and $V(G_2)\setminus V_2$ to $G'$.
		We repeat this process with $G_{i+1}:=G_i-V_i$ until $G_{i+1}$ is empty.
		The process has to stop eventually as each $V_i$ is not empty.
		In this way three sequences are produced: $G=G_1,G_2,\ldots, G_p, G_{p + 1}$, where $G_1$ is $G$ and $G_{p+1}$ is the empty graph; $v_1,v_2,\ldots, v_p$,  and $V_1, V_2, \ldots, V_p$.
		Note that the sequences $\{v_i\}$ and $\{V_i\}$ satisfy property $(ii)$ by construction.
		Observe that this procedure can be clearly performed in time $\mathcal{O}(n^2)$, and for any $r \ge 2$, $m + k = t_r(n) = \Theta(n^2)$, thus the algorithm takes time $\mathcal{O}(m + k)$.
		
		Clearly, $G'$ is a complete $p$-partite graph with parts $V_1, V_2, \ldots, V_p$ as in $G'$ we added all edges between $V_i$ and $V(G_i)\setminus V_i=(V_{i+1}\cup V_{i+2}\cup\ldots\cup V_p)$ for each $i\in \{1,\dots, p\}$ and never added an edge between two vertices in the same $V_i$.
		Since a $p$-partite graph is always $K_{p + 1}$-free, by \Cref{thm:turan} $|E(G')|\le t_p(n)$.
		
		\begin{claim}
			$|E(G')|-|E(G)|\ge \sum_{i=1}^p |E(G[V_i])|$ and for each $u \in V(G)$, $\deg_G(u)\le \deg_{G'}(u)$.
		\end{claim}
		\begin{proof}[of Claim]
			For each $i \in \{1,\dots, p\}$, denote by $E_i$ the edges of $G'$ added in the $i$-{th} step of the construction.
			Formally, $E_i=V_i\times (V_{i+1}\cup V_{i+2}\cup\ldots\cup V_p)$ for $i < p$ and $E_p = \emptyset$.
			We aim to show that $|E_i|-|E(G_i)\setminus E(G_{i + 1})|\ge |E(G[V_i])|$.
			The first part of the claim will follow as $|E(G')|=\sum_{i=1}^{p} |E_i|$ and $|E(G)|=\sum_{i=1}^p |E(G_i)\setminus E(G_{i + 1})|$.
			
			Denote by $d_i$ the degree of $v_i$ in $G_i$.
			Since $N_{G_i}(v_i)=(V_{i+1}\cup V_{i+2}\cup\ldots\cup V_p)$, $|E_i|=d_i|V_i|$.
			As $v_i$ is a maximum-degree vertex in $G_i$, $d_i\ge \deg_{G_i}(u)$ for every $u\in V_i$, so $|E_i|\ge \sum_{u\in V_i}\deg_{G_i}(u)$.
			Recall that $G_{i + 1} =G_i-V_i$.
			Then 
			\begin{align*}
			|E(G_i)\setminus E(G_{i + 1})|=&\sum_{u\in V_i}\deg_{G_i}(u)-|E(G_i[V_i])|=\sum_{u\in V_i}\deg_{G_i}(u)-|E(G[V_i])|\\
			\le &|E_i|-|E(G[V_i])|,
			\end{align*}
			and the first part of the claim follows.
			
			To show the second part, note that for a vertex $u \in V_i$, $\deg_G(u)\le \sum_{j=1}^{i-1}|V_j|+\deg_{G_i}(u)$.
			On the other hand, $u$ is adjacent to every vertex in $V_1\cup V_2\cup\cdots \cup V_{i-1} \cup V_{i+1} \cup\cdots \cup V_p$ in $G'$.
			We have already seen that $|V_{i+1} \cup\cdots \cup V_p|\ge \deg_{G_i}(u)$.
			Thus, $\deg_G(u)\le \deg_{G'}(u)$.
			Proof of the claim is complete.
		\qed\end{proof}
		
		The claim yields  that $|E(G)|\le t_p(n)$, so $t_p(n)\ge t_r(n)-k$.
		By  \Cref{thm:turan}, we have that $t_i(n)>t_{i-1}(n)$, as $T_{i-1}(n)$ is distinct from $T_i(n)$, so $t_i(n)\ge t_{i-1}(n)+1$ for every $i \in [n]$.
		Hence if $r\ge p$ then $k\ge t_r(n)-t_p(n)\ge r-p$.
		It concludes the proof of $(i)$.
		
		It is left to prove  $(iii)$, i.e. that  $|E(G)\triangle E(G')|\le 3k$ under assumption $p\le r$.
		First note that $E(G)\setminus E(G')=\bigcup E(G[V_i])$.
		Second, since $|E(G')|\le t_p(n)\le t_r(n)$ and $|E(G)|\ge t_r(n)-k$, $|E(G')|-|E(G)|\le k$.
		By Claim, we have that $|E(G')|-|E(G)| \ge \sum|E(G[V_i])|.$
		Finally  
		\begin{align*}
		|E(G)\triangle E(G')|=&|E(G')|-|E(G)|+2|E(G)\setminus E(G')|\\=&|E(G')|-|E(G)|+2\sum |E(G[V_i])|\le 3k.
		\end{align*}
		
		By Claim, each vertex covered by $E(G)\setminus E(G')$ is covered by $E(G')\setminus E(G)$.
		The total size of these edge sets is at most $3k$, while $|E(G')\setminus E(G)|-|E(G)\setminus E(G')|=|E(G')|-|E(G)|\le k$.
		Hence, the size of $|E(G)\setminus E(G')|$ is at most $2k$.
		This concludes the proof of $(iii)$ and of the lemma.
	\qed\end{proof}
	%
	
	We are ready to prove our main algorithmic result.  Let us recall that we seek a clique of size $\ell$ in an $n$-vertex graph with $t_r(n)-k$ edges, and that 
	$\tau =\max\{\ell-r,0\}$.

	\begin{theorem}\label{theorem:kernel}
		\probTuranClique~with $\tau\in\{0,1\}$ admits an $\mathcal{O}(n + m)$-time compression into \textsc{Clique} on at most $5k$ vertices.
	\end{theorem}
	\begin{proof}
		Let $(G,r,k,\ell)$ be the input instance of \probTuranClique. If $r < 2$ or $n \le 5k$, a trivial compression is returned.
		Apply the algorithm of  \Cref{lemma:partition} to $(G, r, k, \ell)$ and obtain the partition $V_1, V_2, \ldots, V_p$.
		Observe that this takes time $\mathcal{O}(m + k) = \mathcal{O}(n + m)$ since $n > 5k$.
		By the second property of \Cref{lemma:partition}, $v_1, v_2, \ldots, v_p$ induce a clique in $G$, so if $p\ge \ell$ we conclude that $(G,r,k,\ell)$ is a yes-instance. Formally, the compression returns a trivial yes-instance of \textsc{Clique} in this case.
		
		We now have that $r-k \le p \le r$.
		Then the edit distance between $G$ and the complete $p$-partite graph $G'$ with parts $V_1, V_2, \ldots, V_p$ is at most $3k$.
		Denote by $X$ the set of vertices covered  by $E(G)\triangle E(G')$.
		Denote $R=E(G')\setminus E(G)$ and $A=E(G)\setminus E(G')$.
		We know that $|R|+|A|\le 3k$, $|R|\le 2k$ and $|R|\ge |A|$.
		By \Cref{lemma:partition}, $R$ covers all vertices in $X$, so $|X|\le 2|R|$.
		
		Clearly, $(G,r,k,\ell)$ as an instance of \probTuranClique is equivalent to an instance $(G,\ell)$ of \textsc{Clique}.
		We now apply the following two reduction rules exhaustively to $(G,\ell)$.
		Note that these rules are an adaption of the well-known two reduction rules for the general case of \textsc{Clique} (see, e.g.,\ \cite{XiaoN17}).
		Here the adapted rules employ the partition $V_1, V_2, \ldots, V_p$ explicitly.
		
		\begin{rrule}\label{rrule:remove_not_touched}
			If there is $i \in [p]$ such that $V_i \not\subseteq X$ and $V_i$ is independent in $G$, remove $V_i$ from $G$ and reduce $\ell$ by one.
		\end{rrule}
		
		\begin{rrule}\label{rrule:only_one}
			For each $i \in [p]$ with $|V_i\setminus X|> 1$, remove all but one vertices in $V_i\setminus X$ from $G$.
		\end{rrule}

		Since the reduction rules are applied independently to parts $V_1$, $V_2$, \dots, $V_p$, and each rule is applied to each part at most once, clearly this can be performed in linear time. We now argue that these reduction rules always produce an equivalent instance of \textsc{Clique}.
		
		\begin{claim}
			\Cref{rrule:remove_not_touched} and \Cref{rrule:only_one} are safe.
		\end{claim}
			\begin{proof}
		For \Cref{rrule:remove_not_touched}, note that there is a vertex $v \in V_i\setminus X$  such that $N_G(v)=N_G(V_i)=V(G)\setminus V_i$.
		Since $V_i$ is independent, for any vertex set $C$ that induces a clique in $G$, we have $|C\cap V_i|\le 1$.
		On the other hand, if $C\cap V_i=\emptyset$, $C\cup \{v\}$ also induces a clique in $G$ as $C \subseteq N_G(v)$.
		Hence, any maximal clique in $G$ contains exactly one vertex from $V_i$, so \Cref{rrule:remove_not_touched} is safe.
		
		To see that \Cref{rrule:only_one} is safe, observe that $N_G(u)=N_G(v)$ for any two vertices $u,v\in V_i\setminus X$.
		Then no clique contains both $u$ and $v$, and if $C\ni v$ induces a clique in $G$, $C\setminus \{v\}\cup \{u\}$ also induces a clique in $G$ of the same size.
		Hence, $v$ can be safely removed from $G$ so \Cref{rrule:only_one} is safe.
	\qed\end{proof}
		
		It is left to upperbound the size of $G$ after the exhaustive application of reduction rules.
		In this process, some parts among $V_1, V_2,\ldots, V_p$ are removed from $G$.
		W.l.o.g. assume that the remaining parts are $V_1, V_2, \ldots, V_t$ for some $t \le p$.
		Note that parts that have no common vertex with $X$ are eliminated by \Cref{rrule:remove_not_touched}, so $t \le |X|$.
		On the other hand, by \Cref{rrule:only_one}, we have $|V_i\setminus X|\le 1$ for each $i \in [t]$.
		
		Consider $i \in [t]$ with $|V_i\setminus X|=1$.
		By \Cref{rrule:remove_not_touched}, $G[V_i]$ contains at least one edge.
		Since $V_i$ is independent in $G'$, $E(G[V_i])\subseteq A$.
		Hence, the number of $i \in [t]$ with $|V_i\setminus X|=1$ is at most $|A|$.
		We obtain  
		\begin{align*}
		|V(G)|=&\sum_{i=1}^t |V_i|=\sum_{i=1}^t |V_i\cap X|+\sum_{i=1}^t |V_i\setminus X|\\
		\le &|X|+|A|\le 2|R|+|A|\le |R|+(|R|+|A|)\le 5k.
		\end{align*}
		We obtained an instance of \textsc{Clique} that is equivalent to $(G,r,k,\ell)$ and contains at most $5k$ vertices.
		The proof is complete.
	\qed\end{proof}
	
	Combining the polynomial compression of \Cref{theorem:kernel} with the algorithm of Xiao and Nagamochi \cite{XiaoN17} for \textsc{Independent Set} running in $\mathcal{O}(1.1996^n)$, we obtain the following.
	
	\begin{corollary}\label{corollary:fpt}
		\probTuranClique~with $\tau\le 1$ is solvable in time  $2.49^{k}\cdot(n + m)$.
	\end{corollary}
	\begin{proof}
		Take a given instance of \probTuranClique and compress it into an equivalent instance $(G,\ell)$ of \textsc{Clique} with $|V(G)|\le 5k$.
		Clearly, $(G,|V(G)|-\ell)$ is an instance of \textsc{Independent Set} equivalent to $(G,\ell)$.
		Use the algorithm from \cite{XiaoN17} to solve this instance in $\mathcal{O}(1.1996^{|V(G)|})$ running time.
		Since $1.1996^5<2.49$, the running time of the whole algorithm is bounded by $2.49^k\cdot\polyn$.
	\qed\end{proof}

	\subsection{Looking for larger cliques}
	
	In this subsection we consider the situation when  $\tau >  1$.   As we will see in \Cref{thm:TuranNPC}, an FPT algorithm is unlikely in this case, unless we take a stronger parameterization. 
	Here we  show that \probTuranClique is FPT parameterized by  $\tau+\xi+k$. Recall that $\xi= \lfloor\frac{n}{r}\rfloor $.
	\Cref{thm:TuranNPC} argues that this particular choice of the parameter is necessary. 
	
	First, we show that the difference between $t_\ell(n)$ and $t_r(n)$ can be bounded in terms of $\tau$ and $\xi$. This will allow us to employ \Cref{theorem:kernel} for the new FPT algorithm by a simple change of the parameter.
	The proof of the next lemma is done via a careful counting argument.
	
	\begin{lemma}\label{lemma:turan_diff}
		Let $n, r,\ell$ be three positive integers with $r< \ell \le n$. Let $\xi=\lfloor\frac{n}{r}\rfloor$ and $\tau=\ell-r$. 
		Then for  $\tau =\O(r)$, $t_\ell(n)-t_r(n)=\Theta(\tau\xi^2)$.
		
	\end{lemma}
	\begin{proof}
	Throughout the proof, we assume $\xi=\frac{n}{r}$ since this does not influence the desired $\Theta$ estimation.
	Let $s_r$ be the remainder in the division of $n$ by $r$ and  $s_\ell$ be the remainder in the division of $n$ by $\ell$.
	By \Cref{lemma:turans_edges},  
	\begin{equation}\label{eq:turan_delta}
		t_{\ell}(n)-t_r(n)= \frac{\tau n^2}{2r\ell}+\left(\frac{s_r}{2}\cdot\left(1-\frac{s_r}{r}\right)-\frac{s_{\ell}}{2}\cdot \left(1-\frac{s_\ell}{\ell}\right)\right).
	\end{equation}
	
	The first summand in \eqref{eq:turan_delta}    is $\Theta(\xi^2 \tau)$.  
	Indeed, since   $\tau =\O(r)$ we have 
	\begin{eqnarray}\frac{\tau n^2}{2r\ell}=\frac{\tau}{2}\cdot\frac{n}{r}\cdot \frac{n}{r+\tau}=\frac{\xi^2\tau}{2}\cdot \frac{r}{r+\tau}=\Theta(\xi^2 \tau).
		\label{eq:firstpart}\end{eqnarray}

	For the second summand,
	\begin{eqnarray}
		\frac{s_r}{2}\cdot\left(1-\frac{s_r}{r}\right)&-&\frac{s_{\ell}}{2}\cdot \left(1-\frac{s_\ell}{\ell}\right)\\
		&=&\frac{\ell s_r(r-s_r)-rs_\ell(\ell-s_\ell)}{2r\ell}=\frac{(rs_\ell^2-\ell s_r^2)+r\ell(s_r -s_\ell )}{2r\ell} \nonumber \\
		&=&\frac{(rs_\ell^2-r s_r^2-\tau s_r^2)+r\ell(s_r -s_\ell )}{2r\ell}\\ &=&\frac{r(s_\ell-s_r)(s_\ell+s_r)+r\ell (s_r-s_\ell)}{2r\ell}-\frac{\tau s_r^2}{2r\ell} \nonumber \\
		& = & \frac{(s_r-s_\ell)(\ell-(s_\ell+s_r))}{2\ell}-\frac{\tau s_r^2}{2r\ell}. \label{eq:right_part}
	\end{eqnarray}
	Since $n=\lfloor \frac{n}{\ell} \rfloor \cdot \ell + s_\ell$, we have that 
	$$s_r\equiv \left\lfloor \frac{n}{\ell} \right\rfloor \cdot \ell + s_\ell\pmod{r},$$
	and 
	$$s_r\equiv \left\lfloor \frac{n}{\ell}\right\rfloor\cdot (r+\tau) + s_\ell  \pmod{r}.$$
	Hence,
	$$s_r-s_\ell\equiv \left\lfloor \frac{n}{\ell}\right\rfloor\cdot \tau \pmod{r}.$$	
	By definition  $s_r<r$, thus we get from the above that $s_r-s_\ell \le \lfloor\frac{n}{\ell}\rfloor \cdot \tau \le \xi\tau.$
	
	Analogously, $$s_\ell-s_r\equiv \left\lfloor\frac{n}{r}\right\rfloor\cdot (-\tau) \pmod{\ell}$$
	Since $s_\ell-s_r> -r > -\ell$, we have that $s_\ell-s_r\ge \lfloor\frac{n}{r}\rfloor \cdot (-\tau)\ge -\xi\tau.$
	Therefore $|s_\ell-s_r|\le \xi\tau$.
	It is easy to see that $|\ell-(s_\ell+s_r)|\le \ell+(s_\ell+s_r)\le 3\ell$.
	Finally, $\frac{\tau s^2_r}{2r\ell}$ is non-negative and is upper bounded by $\frac{\tau r^2}{2r\ell}\le \frac{\tau}{2}$.
	Thus, the absolute value of (\ref{eq:right_part}), 
	is at most
	$$\frac{\xi\tau\cdot3\ell}{2\ell}+\frac{\tau}{2}=\mathcal{O}(\xi\tau).$$
	
	By putting together \eqref{eq:firstpart}	and \eqref{eq:right_part}, we conclude  that  $t_\ell(n)-t_r(n)=\Theta(\xi^2\tau) + \mathcal{O}(\xi\tau)=\Theta(\xi^2\tau)$.
\qed\end{proof}

	The following compression algorithm is 
	a   corollary of \Cref{lemma:turan_diff} and \Cref{theorem:kernel}.
	It provides a compression of size linear in $k$ and $\tau$.  
	
	\begin{theorem}\label{thm:tauxik}
		\probTuranClique~admits a  compression into \textsc{Clique} on  $\mathcal{O}(\tau\xi^2+k)$ vertices.
	\end{theorem}
	\begin{proof}
		Let $(G,k,r,\ell)$ be the given instance of \probTuranClique.
		If $\ell\le r+1$, then the proof follows from  \Cref{theorem:kernel}.
		Otherwise, reduce $(G,k,r,\ell)$ to an equivalent instance   $(G,k+t_\ell(n)-t_r(n),\ell,\ell)$ of \probTuranClique just by modifying the parameters.
		This is a valid instance since $|E(G)|\ge t_r(n)-k\ge t_\ell(n)-(t_\ell(n)+t_r(n)+k)$.
		Denote $k'=k+(t_\ell(n)-t_r(n))$.
		By \Cref{lemma:turan_diff}, $k'=k+\mathcal{O}(\tau\xi^2)$.
		Apply polynomial compression of \Cref{theorem:kernel} to $(G,k',\ell,\ell)$ into \textsc{Clique} with $\mathcal{O}(k')$, i.e.\ $\mathcal{O}(\tau\xi^2+k)$, vertices.
	\qed\end{proof}

	Pipelined with a brute-force algorithm computing a maximum independent set in time $\O(2^n)$, \Cref{thm:tauxik} yields  the following corollary. 
	\begin{corollary}\label{cor:xialgorithm}
		\probTuranClique  is solvable in time   $2 ^{\O(\tau\xi^2+k)}\cdot (n + m)$.
	\end{corollary}
	
	\subsection{Independent set above Tur\'{a}n's bound} 
	
	Another interesting application of  \Cref{thm:tauxik} concerns   computing 
	\textsc{Independent Set} in graphs of small average degree. Recall that 
	Tur\'{a}n's  theorem, when applied to the complement $\overline{G}$ of a graph $G$, yields a  bound
	\[
	\alpha({G})\geq \frac{n}{d + 1}.
	\]
	Here  $\alpha(G)$ is the size of the largest independent set in $G$ (the independence number of $G$), and $d$ is the average vertex degree of $G$.  Then in 
	\probTuranIS, 
	 the task is for an $n$-vertex graph $G$ and   positive integer $t$ to decide whether there is  an independent set  of size at least $\frac{n}{d + 1}+t$ in $G$.
	
	%
	%
	\Cref{thm:tauxik} implies a compression of \probTuranIS into \textsc{Independent Set}. In other words, we give a polynomial time algorithm that for an instance $(G, t)$ of  \probTuranIS constructs an equivalent instance $(G', p)$ of \textsc{Independent Set} with at most $\mathcal{O}(td^2)$ vertices.  That is, the graph $G$ has an independent set of size at least $\frac{n}{d+1} +t$ if and only if $G'$ has an independent set of size $p$.

	\begin{corollary}\label{cor:comprBrookcompr}
		\probTuranIS~admits a  compression into \textsc{Independent Set} on  $\mathcal{O}(td^2)$ vertices.
	\end{corollary} 
	\begin{proof}
		%
		For simplicity, let us assume that $n$ is divisible by   $d +1$. (For arguments here this assumption does not make an essential difference.)
		We select $r=\frac{n}{d+1}$,  $\tau =t$, and $k=0$. Then $d=\frac{n}{r}-1=\xi -1$. The graph $\overline{G}$ has at most $nd/2$ edges, hence $G$ has at least $\frac{n(n-1)}{2}- nd/2=\frac{n(n-1)}{2}- n(\xi-1)/2 \geq t_r(n)$ edges,  see \Cref{lemma:turans_edges}.  
		An independent set of size $\frac{n}{d+1} +t$ in graph  $\overline{G}$,  corresponds   in graph $G$  to a clique of size $r+t$. Since \Cref{thm:tauxik} provides  compression into a \textsc{Clique} with 
		$\mathcal{O}(\tau\xi^2+k)=\mathcal{O}(\tau\xi^2)$ vertices, for independent set and graph $\overline{G}$ this corresponds to a compression into an instance of \textsc{Independent Set} with  $\mathcal{O}(td^2)$ vertices. 
	\qed\end{proof}
	
	By \Cref{cor:comprBrookcompr}, we obtain the following corollary.
	
	\begin{corollary}\label{cor:comprBrookalgo}
		\probTuranIS  is solvable in time   $2 ^{\O(td^2)}\cdot n^2$.
	\end{corollary} 
	%

	\section{Lower bounds}\label{sec:lowerbounds}
	
	In this section, we investigate how the algorithms above are complemented by hardness results.
	First, observe that $k$ has to be restricted, otherwise the \probTuranClique problem is not any different from \textsc{Clique}.
	In fact, reducing from \textsc{Independent Set} on sparse graphs, one can show that there is no $2^{o(k)}$-time algorithm for \probTuranClique even when $\tau \le 1$. (The formal argument is presented in \Cref{thm:ETH}.) This implies that the $2^{\mathcal{O}(k)}$-time algorithm given by \Cref{corollary:fpt} is essentially tight.
	
	Also, the difference between $r$ and $\ell$ has to be restricted, as it can be easily seen that \probTuranClique admits no $n^{o(\ell)}$-time algorithm even when $k = 0$, assuming ETH. This is observed simply by considering the special case of \probTuranClique where $r = 1$, there the only restriction on $G$ is that $|E(G)| \ge t_r(n) - k = 0$, meaning that the problem is as hard as \textsc{Clique}.
	However, \Cref{thm:TuranNPC} shows that even for any fixed $\tau \ge 2$ and $k = 0$ \probTuranClique is \NP-complete.
	This motivates \Cref{thm:tauxik}, where the exponential part of the running time has shape $2^{\mathcal{O}(\tau \xi^2 k)}$.
	In the rest of this section, we further motivate the running time of \Cref{thm:tauxik}. First, in \Cref{thm:TuranNPC} we show that not only setting $\tau$ and $k$ to constants is not sufficient to overcome \NP-hardness, but also that the same holds for any choice of two parameters out of $\{\tau, \xi, k\}$.
	
	
	\begin{theorem}\label{thm:TuranNPC}
		\probTuranClique~is \NP-complete. Moreover, it remains \NP-complete in each of the following cases
		\begin{itemize}
			\item[$(i)$] for any fixed $\xi\ge 1$ and $\tau=0$;
			\item[$(ii)$] for any fixed $\xi\ge 1$ and $k=0$;
			\item[$(iii)$] for any fixed $\tau \ge 2$ and $k=0$.
		\end{itemize}
	\end{theorem}
	\begin{proof}
		Towards proving $(i)$ and $(ii)$, we provide a reduction from \textsc{Clique}. Let $\xi\ge 1$ be a fixed constant.
		Let $(G,\ell)$ be a given instance of \textsc{Clique} and  let $n=|V(G)|$.
		We assume that $\ell\ge \xi$, otherwise we can solve $(G,\ell)$ in polynomial time.
		Construct a graph $G'$ from $G$ as follows.
		Start from $G'=G$ and $\ell'=\ell$.
		Then add $\max\{\xi \ell- n,0\}$ isolated vertices to $G'$. Note that $(G,k)$ and $(G',k')$ are equivalent and $|V(G')|\ge \xi \ell'$.
		If we have $\xi \ell' \le |V(G')|<(\xi+1)\ell'$, we are done with the construction of $G'$.
		Otherwise, repeatedly add a universal vertex to $G'$, increasing $\ell'$ by one, so $|V(G')|-(\xi+1)\ell'$ decreases by $\xi$ each time.
		We repeat this until $|V(G')|$ becomes less than $(\xi+1)\ell'$.
		Since the gap between $\xi \ell'$ and $(\xi+1)\ell'$ is at least $\xi$ at any moment, we derive that $\xi \ell' \le |V(G')| < (\xi+1)\ell'$.
		The construction of $G'$ is complete. Note that$(G',\ell')$ is an instance of \textsc{Clique} equivalent to $(G,\ell)$.
		We added at most $\max \{n,\xi \ell\}$ vertices to $G'$, hence this is a polynomial-time reduction.
		
		By the above, $\lfloor V(G')/\ell' \rfloor = \xi$, so we can reduce $(G',\ell')$ to an equivalent instance $(G',\ell',\binom{|V(G')|}{2},\ell')$ of \probTuranClique.
		Clearly, this instance has the required fixed value of $\xi$ and $\tau=0$.
		This proves $(i)$.
		For $(ii)$, we use the fact that $t_1(n)=0$ for every $n>0$ and reduce $(G',\ell')$ to $(G',1,0,\ell')$.

		To show  $(iii)$, we need another reduction from \textsc{Clique}.
		Let $\tau \ge 2$ be a fixed integer constant.
		Take an instance $(G,\ell)$ of \textsc{Clique} with $\ell\ge 2\tau$.
		We denote $n=|V(G)|$.
		To construct $G'$ from $G$, we start from a large complete $(\ell-1)$-partite graph with equal-sized parts.
		The size of each part equals $x$, so $|V(G')|=(\ell-1)x$.
		We denote $N=|V(G')|$ and choose the value of $x$ later, for now we only need that $N\ge n$.
		Clearly, $|E(G')|=t_{\ell-1}(N)$ at this point.
		To embed $G$ into $G'$, we select arbitrary $n$ vertices in $G'$ and make them isolated.
		This removes at most $n(\ell-2)x$ edges from $G'$.
		Then we identify these $n$ isolated vertices with $V(G)$ and add edges of $G$ between these vertices in $G'$ correspondingly.
		This operation does not decrease $|E(G')|$.
		This completes the  construction of $G'$.
		Since $G'$ is isomorphic to a complete $(\ell-1)$-partite graph united disjointly with $G$, we have that $(G,\ell)$ and $(G',\ell)$ are equivalent instances of \textsc{Clique}.
		
		We now want to reduce $(G',\ell)$ to an instance $(G',\ell-\tau,0,\ell)$ of \probTuranClique.
		To do so, we need $|E(G')|\ge t_{\ell-\tau}(N)$.
		By \Cref{lemma:turan_diff}, $t_{\ell-1}(N)-t_{\ell-\tau}(N)\ge C \cdot (\tau-1) \cdot \left(\frac{N}{\ell-\tau}\right)^2$ for some 
		constant 
		$C>0$.
		Since $|E(G')|\ge t_{\ell-1}(N)-n(\ell-2)x$, we want to choose $x$ such  that 
		
		
		$$n(\ell-2)x\le C \cdot (\tau-1) \cdot \left(\frac{N}{\ell-\tau}\right)^2.$$
		By substituting $N=(\ell-1)x$,  we derive that $x$ should satisfy 
		
		$$\frac{n}{C}\cdot \frac{(\ell-2)(\ell-\tau)}{(\ell-1)^2}\cdot \frac{\ell-\tau}{\tau-1}\le x.$$
		
		Now simply pick as $x$ the smallest integer that satisfies the above.
		Then $(G',\ell-\tau,0,\ell)$ is an instance of \probTuranClique~that is equivalent to the instance  $(G,k)$ of \textsc{Clique} and is constructed in polynomial time.
	\qed\end{proof}
	
	Now, recall that \Cref{thm:tauxik} gives an \FPT-algorithm for \probTuranClique~that is single-exponential in $\tau\xi^2+k$. 
	The previous theorem argues that all three of $\tau$, $\xi$, $k$ have to be in the exponential part of the running time. However, that result does not say anything about what can be the best possible dependency on these parameters.
	The next \Cref{thm:ETH} aims to give more precise lower bounds based on ETH, in particular it turns out that the dependency on $\tau$ and $k$ cannot be subexponential unless ETH fails.
	First, we need to show the relation between the parameter $\xi$ and the average degree of $\overline G$.
	
	\begin{proposition}\label{obs:avg_degree}
		Let $G$ be an $n$-vertex graph, $r\le n$ be an integer, and denote $\xi=\lfloor\frac{n}{r}\rfloor$.
		Let $\overline{G}$ denote the complement of $G$ and $\overline d$ denote the average degree of $\overline G$.
		Then $\overline d \le \xi$ if $|E(G)|\ge t_r(n)$ and  $|E(G)|\ge t_r(n)$ if $\overline d \le \xi-1$.
	\end{proposition}
		\begin{proof}
		Let $s<r$ be the remainder in the division of $n$ by $r$.
		Then
		\begin{eqnarray*}
			\binom{n}{2}-t_r(n)&=&\frac{n^2}{2}-\frac{n}{2}-\left(1-\frac{1}{r}\right)\cdot\frac{n^2}{2}+\frac{s}{2}\cdot\left(1-\frac{s}{r}\right)\\&=&\frac{n}{2}\cdot \frac{n}{r}-\frac{n}{2}+\frac{s}{2}\cdot \left(1-\frac{s}{r}\right)\\
			&=&\frac{n}{2}\cdot \xi +\frac{n}{2}\cdot \frac{s}{r}-\frac{n}{2}+\frac{s}{2}-\frac{s}{2}\cdot \frac{s}{r}=\frac{n}{2}\cdot \xi-\frac{n-s}{2}\cdot \left(1-\frac{s}{r}\right)=\\
			&=&\frac{n}{2}\cdot \xi-\frac{(n-s)(r-s)}{2r}
			=\frac{n}{2}\cdot \xi - \frac{r-s}{2}\cdot \xi. 
		\end{eqnarray*}
		Since $|E(\overline G)|=\binom{n}{2}-|E(G)|$, one direction is proved: assuming $|E(G)|\ge t_r(n)$, $|E(\overline G)|\le \binom{n}{2}-t_r(n)\le \frac{\xi n}{2}$.
		
		For the other direction, assume that $\overline d \le \xi-1$, i.e.\ $|E(\overline G)|\le \frac{n}{2}\cdot (\xi-1)$.
		As $\binom{n}{2}-t_r(n)\ge \frac{n}{2}\cdot \xi-\frac{n}{2}$, we have $\binom{n}{2}-t_r(n)\ge |E(\overline G)|$.
		Then $|E(G)|\ge t_r(n)$ follows and the proof is complete.
	\qed\end{proof}
	
	We are ready to give lower bounds for algorithms solving \probTuranClique in terms of the parameters $\tau$, $\xi$, and $k$.
	
	\begin{theorem}\label{thm:ETH}
		Unless the Exponential Time Hypothesis fails, for any function $f$ there is no $f(\xi,\tau)^{o(k)}\cdot n^{f(\xi,\tau)}$, $f(\xi,k)^{o(\tau)} \cdot n^{f(\xi,k)}$, or $f(k,\tau)^{o(\sqrt{\xi})}\cdot n^{f(k,\tau)}$ algorithm for \probTuranClique.
	\end{theorem}

			\begin{proof}
		It is well-known that under ETH \textsc{Independent Set} cannot be solved in $2^{o(n)}$ time on instances with linear number of edges \cite{cygan2015parameterized}.
		This is a basis of our proof: we provide several reductions from \textsc{Independent Set} with linear number of edges.
		Note that these reductions mostly repeat the reductions given in the proof of \Cref{thm:TuranNPC} but are different in terms of requirements for $\tau, \xi$ and $k$.
		In fact, we give three polynomial-time algorithms that reduce an instance $(G,q)$ of \textsc{Independent Set}, where $n=|V(G)|$ and $|E(G)|=\mathcal{O}(n)$, to an equivalent instance of \probTuranClique such that
		
		\begin{itemize}
			\item $\tau, \xi$ are constant but $k=\mathcal{O}(n)$;
			\item $\xi, k$ are constant but $\tau=\mathcal{O}(q)$;
			\item $\tau, k$ are constant but $\xi=\mathcal{O}(nq)$.
		\end{itemize}
		
		Clearly, once we show these three reductions, the proof of the theorem is complete.

		For an instance of \textsc{Independent Set} $(G,q)$ we denote $n=|V(G)|$ and $m=|E(G)|$.
		We always assume that the number of edges in $G$ is linear, so $m=\mathcal{O}(n)$ and the average degree of $G$ is not greater than some fixed constant $D$.
		We also assume that $n\ge 2(D+1)$ and $q>2$.
		
		The first reduction takes $(G,q)$ and trivially reduces it to an instance $(\overline G,n,m,q)$ of \probTuranClique.
		Note that $|E(\overline G)|=\binom{n}{2}-m=t_n(n)-m$, so this is a valid instance of the problem.
		For this instance, $\xi=1$ and $\tau=0$ but $k=m=\mathcal{O}(n)$.
		
		The second algorithm reduces $(G,q)$ to an equivalent instance $(\overline G,r,0,q)$, where $r=\lfloor\frac{n}{D+1}\rfloor$.
		As $\lfloor n/r\rfloor-1$ upper bounds the average degree of $G$, by \Cref{obs:avg_degree} we have that $|E(\overline G)|\ge t_r(n)$, so $(\overline G, r,0,q)$ is a valid instance.
		This instance has $k=0$ and $$\xi=\lfloor n/r\rfloor < n\cdot \left(\frac{n}{D+1}-1\right)^{-1}=(D+1)\cdot \frac{n}{n-(D+1)}\le 2(D+1),$$
		but $\tau\le q=\mathcal{O}(q)$.
		
		To show the last reduction, we reduce from the instance $(G,\ell)$ of \textsc{Clique} instead of $(G,q)$ of \textsc{Independent Set}, since constraints on the number of edges in $G$ are not necessary for it.
		Formally it means that we reduce from $(G,q)$ of \textsc{Independent Set} to $(\overline G,q)$ of \textsc{Clique}, and then apply reductions as required.
		Slightly abusing the notion we denote $(\overline G, q)$ by $(G,\ell)$.
		
		We adjust the last reduction from \textsc{Clique} from the proof of \Cref{thm:TuranNPC}.
		Recall that in this reduction we reduce an instance $(G,\ell)$ of \textsc{Clique} to an equivalent instance $(G',\ell)$ of \textsc{Clique} with $|V(G')|=(\ell-1)x$ for some chosen integer $x$.
		For $(G',\ell-\tau,0,\ell)$ to be a valid equivalent instance of \probTuranClique, it is enough that $x$ satisfies
		$$x \ge \frac{n}{C}\cdot \frac{\ell-\tau}{\ell - 1}\cdot \frac{\ell-\tau}{\tau-1},$$
		where the fixed constant $C>0$ comes from \Cref{lemma:turan_diff}.
		
		To show the third reduction, we pick $\tau:=2$.
		Then we choose $x:=\lceil n\ell/C \rceil$, so $(G',\ell-\tau,0,\ell)$ is a valid instance of \probTuranClique~equivalent to $(G,\ell)$.
		This instance has $k=0$ and $\tau=2$, but $\xi\le |V(G')|/\ell<x=\mathcal{O}(n\ell)$.
		\qed\end{proof}
	
	The first part of \Cref{thm:ETH} lets us observe that our $2.49^k \cdot (n +m)$-time algorithm for \probTuranClique with $\tau \le 1$ is essentially tight.
	
	\begin{corollary}\label{cor:single_tight}
		Assuming ETH, there is no $2^{o(k)} \cdot \polyn$ algorithm for \probTuranClique with $\ell \le r + 1$.
	\end{corollary}

	\section{Conclusion}\label{sec:conclustion}
	We conclude by summarizing   natural questions left open by our work.  \Cref{thm:ETH} 
	rules out (unless ETH fails) algorithms with running times subexponential in $\tau $ and $k$. However, when it comes to $\xi$,  the  dependency
	in the   upper bound of \Cref{cor:xialgorithm}  is  $2^{\mathcal{O}(\tau\xi^2+k)}  \cdot\polyn$, while \Cref{thm:ETH} only rules out the running time of $f(k,\tau)^{o(\sqrt{\xi})}\cdot n^{f(k,\tau)}$ under ETH.
	Thus, whether the correct dependence in $\xi$ is single-exponential or subexponential, is left open.
	%
	%
	%
	%
	%
	%
	Similarly,  the question whether \probTuranClique admits a compression into \textsc{Clique} whose size is linear in  
	$\xi$, $\tau$, and $k$, is open. A weaker variant of this question (for the case $k=0$) for \probTuranIS, whether it admits a  compression or kernel linear in $d$ and in $t$,  is also open.
	
	%
	%
	%
	%
	%
	
	\bibliography{book_kernels_fvf}  

\begin{thebibliography}{10}

\bibitem{AlonGKSY10}
{\sc N.~Alon, G.~Gutin, E.~J. Kim, S.~Szeider, and A.~Yeo}, {\em Solving
  {MAX}-$r$-{SAT} above a tight lower bound}, Algorithmica, 61 (2011),
  pp.~638--655.

\bibitem{BodlaenderFLPST16}
{\sc H.~L. Bodlaender, F.~V. Fomin, D.~Lokshtanov, E.~Penninkx, S.~Saurabh, and
  D.~M. Thilikos}, {\em ({M}eta) {K}ernelization}, J. ACM, 63 (2016),
  pp.~44:1--44:69.

\bibitem{BodlaenderJK14}
{\sc H.~L. Bodlaender, B.~M.~P. Jansen, and S.~Kratsch}, {\em Kernelization
  lower bounds by cross-composition}, SIAM Journal on Discrete Mathematics, 28
  (2014), pp.~277--305.

\bibitem{Brooks41}
{\sc L.~R. Brooks}, {\em On colouring the nodes of a network}, Proc. Cambridge
  Philos. Soc., 37 (1941), pp.~194--197.

\bibitem{CrowstonJMPRS13}
{\sc R.~Crowston, M.~Jones, G.~Muciaccia, G.~Philip, A.~Rai, and S.~Saurabh},
  {\em Polynomial kernels for lambda-extendible properties parameterized above
  the {P}oljak-{T}urzik bound}, in IARCS Annual Conference on Foundations of
  Software Technology and Theoretical Computer Science (FSTTCS), vol.~24 of
  Leibniz International Proceedings in Informatics (LIPIcs), Dagstuhl, Germany,
  2013, Schloss Dagstuhl--Leibniz-Zentrum fuer Informatik, pp.~43--54.

\bibitem{cygan2015parameterized}
{\sc M.~Cygan, F.~V. Fomin, L.~Kowalik, D.~Lokshtanov, D.~Marx, M.~Pilipczuk,
  M.~Pilipczuk, and S.~Saurabh}, {\em Parameterized Algorithms}, Springer,
  2015.

\bibitem{DemaineFHT05jacm}
{\sc E.~D. Demaine, F.~V. Fomin, M.~Hajiaghayi, and D.~M. Thilikos}, {\em
  Subexponential parameterized algorithms on graphs of bounded genus and
  {$H$}-minor-free graphs}, J. ACM, 52 (2005), pp.~866--893.

\bibitem{DowneyF99}
{\sc R.~G. Downey and M.~R. Fellows}, {\em Parameterized complexity},
  Springer-Verlag, New York, 1999.

\bibitem{DvorakL16}
{\sc Z.~Dvor{\'{a}}k and B.~Lidick{\'{y}}}, {\em Independent sets near the
  lower bound in bounded degree graphs}, in Proceedings of the 34th
  International Symposium on Theoretical Aspects of Computer Science (STACS),
  vol.~66 of Leibniz International Proceedings in Informatics (LIPIcs), Schloss
  Dagstuhl - Leibniz-Zentrum fuer Informatik, 2017, pp.~28:1--28:13.

\bibitem{DvorakM17}
{\sc Z.~Dvor{\'{a}}k and M.~Mnich}, {\em Large independent sets in
  triangle-free planar graphs}, {SIAM} J. Discret. Math., 31 (2017),
  pp.~1355--1373.

\bibitem{MR307975}
{\sc P.~Erd\H{o}s}, {\em On the graph theorem of {T}ur\'{a}n}, Mat. Lapok, 21
  (1970), pp.~249--251.

\bibitem{fomin_et_al:LIPIcs:2019:11168}
{\sc F.~V. Fomin, P.~A. Golovach, D.~Lokshtanov, F.~Panolan, S.~Saurabh, and
  M.~Zehavi}, {\em Going far from degeneracy}, SIAM J. Discrete Math., 34
  (2020), pp.~1587--1601.

\bibitem{fomin2019kernelization}
{\sc F.~V. Fomin, D.~Lokshtanov, S.~Saurabh, and M.~Zehavi}, {\em
  Kernelization: theory of parameterized preprocessing}, Cambridge University
  Press, 2019.

\bibitem{GargP16}
{\sc S.~Garg and G.~Philip}, {\em Raising the bar for vertex cover:
  Fixed-parameter tractability above a higher guarantee}, in Proceedings of the
  Twenty-Seventh Annual {ACM-SIAM} Symposium on Discrete Algorithms (SODA),
  {SIAM}, 2016, pp.~1152--1166.

\bibitem{DBLP:journals/mst/GutinKLM11}
{\sc G.~Gutin, E.~J. Kim, M.~Lampis, and V.~Mitsou}, {\em Vertex cover problem
  parameterized above and below tight bounds}, Theory of Computing Systems, 48
  (2011), pp.~402--410.

\bibitem{GutinIMY12}
{\sc G.~Gutin, L.~van Iersel, M.~Mnich, and A.~Yeo}, {\em Every ternary
  permutation constraint satisfaction problem parameterized above average has a
  kernel with a quadratic number of variables}, J. Computer and System
  Sciences, 78 (2012), pp.~151--163.

\bibitem{GutinP16}
{\sc G.~Z. Gutin and V.~Patel}, {\em Parameterized traveling salesman problem:
  Beating the average}, {SIAM} J. Discrete Math., 30 (2016), pp.~220--238.

\bibitem{GutinRSY07}
{\sc G.~Z. Gutin, A.~Rafiey, S.~Szeider, and A.~Yeo}, {\em The linear
  arrangement problem parameterized above guaranteed value}, Theory Comput.
  Syst., 41 (2007), pp.~521--538.

\bibitem{MR1687331}
{\sc J.~H\aa~stad}, {\em Clique is hard to approximate within
  {$n^{1-\epsilon}$}}, Acta Math., 182 (1999), pp.~105--142.

\bibitem{ImpagliazzoPZ01}
{\sc R.~Impagliazzo, R.~Paturi, and F.~Zane}, {\em Which problems have strongly
  exponential complexity}, J. Computer and System Sciences, 63 (2001),
  pp.~512--530.

\bibitem{DBLP:conf/wg/Jansen0N19}
{\sc B.~M.~P. Jansen, L.~Kozma, and J.~Nederlof}, {\em Hamiltonicity below
  {D}irac's condition}, in Proceedings of the 45th International Workshop on
  Graph-Theoretic Concepts in Computer Science (WG), vol.~11789 of Lecture
  Notes in Computer Science, Springer, 2019, pp.~27--39.

\bibitem{Karp72}
{\sc R.~M. Karp}, {\em Reducibility among combinatorial problems}, in
  Complexity of computer computations, Plenum Press, New York, 1972,
  pp.~85--103.

\bibitem{korandi2021exact}
{\sc D.~Kor{\'a}ndi, A.~Roberts, and A.~Scott}, {\em Exact stability for
  tur{\'a}n’s theorem}, Advances in Combinatorics,  (2021), p.~31079.

\bibitem{LokshtanovNRRS14}
{\sc D.~Lokshtanov, N.~S. Narayanaswamy, V.~Raman, M.~S. Ramanujan, and
  S.~Saurabh}, {\em Faster parameterized algorithms using linear programming},
  {ACM} Trans. Algorithms, 11 (2014), pp.~15:1--15:31.

\bibitem{MahajanR99}
{\sc M.~Mahajan and V.~Raman}, {\em Parameterizing above guaranteed values:
  {MaxSat} and {MaxCut}}, J. Algorithms, 31 (1999), pp.~335--354.

\bibitem{MahajanRS09}
{\sc M.~Mahajan, V.~Raman, and S.~Sikdar}, {\em Parameterizing above or below
  guaranteed values}, J. Computer and System Sciences, 75 (2009), pp.~137--153.

\bibitem{PilipczukS21}
{\sc M.~Pilipczuk and S.~Siebertz}, {\em Kernelization and approximation of
  distance-r independent sets on nowhere dense graphs}, Eur. J. Comb., 94
  (2021), p.~103309.

\bibitem{MR18405}
{\sc P.~Tur\'{a}n}, {\em Eine {E}xtremalaufgabe aus der {G}raphentheorie}, Mat.
  Fiz. Lapok, 48 (1941), pp.~436--452.

\bibitem{XiaoN17}
{\sc M.~Xiao and H.~Nagamochi}, {\em Exact algorithms for maximum independent
  set}, Inf. Comput., 255 (2017), pp.~126--146.

\end{thebibliography}
\end{document}